\documentclass[12pt]{article}

\usepackage{amsmath,amssymb,amsthm}

\newtheorem{thm}{Theorem}
\newtheorem{cor}{Corollary}
\newtheorem{prop}{Proposition}

\title{Perfect mixed codes from generalized Reed-Muller codes}

\date{}

\author{Alexander M. Romanov
\thanks {Sobolev Institute of Mathematics, 630090 Novosibirsk, Russia.}}

\begin{document}
\maketitle

\begin{abstract}
In this paper, we propose a new method for constructing $1$-perfect mixed codes in the Cartesian product
$\mathbb{F}_{n} \times  \mathbb{F}_{q}^n$,
where $\mathbb{F}_{n}$ and $\mathbb{F}_{q}$ are finite fields of orders $n = q^m$ and  $q$.
We consider  generalized Reed-Muller codes of length $n = q^m$ and order  $(q - 1)m - 2$.
Codes whose parameters are the same as the parameters of generalized Reed-Muller codes are called  \emph{ Reed-Muller-like codes}.
The construction we propose is based on partitions of  distance-2 MDS codes  into Reed-Muller-like codes of order $(q - 1)m - 2$.
We construct a set of  $q^{q^{cn}}$  nonequivalent 1-perfect mixed codes in the Cartesian  product $\mathbb{F}_{n} \times \mathbb{F}_{q}^{n}$,
where the constant $c$ satisfies $c < 1$, $n = q^m$ and $m$ is a  sufficiently large positive integer.
We also prove that each  $1$-perfect mixed code in the Cartesian product
$\mathbb{F}_{n} \times  \mathbb{F}_{q}^n$ corresponds
to a certain partition  of a distance-2 MDS code  into Reed-Muller-like codes of order $(q - 1)m - 2$.

\indent

\end{abstract}

\section{Introduction}\label{sec:intr}


Let $n$ be a positive integer and let $q_i$ be a power of a prime $p_i$, $i = 1, 2, \ldots,  n$.
Consider the Cartesian product
$V_n := \mathbb{F}_{1}\times \mathbb{F}_{2} \times \cdots \times \mathbb{F}_{n}$,
where $\mathbb{F}_{i}$ is a finite field of order $q_i$.
An arbitrary non-empty subset ${\mathcal C}$ of $V_n$ is called a \emph{code} of length $n$.
If  $\mathbb{F}_{1}, \mathbb{F}_{2}, \ldots, \mathbb{F}_{n}$ are fields of different orders,
then the code ${\mathcal C}$ is called a \emph{mixed code}.
If the orders of all finite fields  $\mathbb{F}_{1}, \mathbb{F}_{2}, \ldots, \mathbb{F}_{n}$ are equal to $q$,
then $V_n$ is an $n$-dimensional vector space over $\mathbb{F}_q$
and the code ${\mathcal C}$ is called a \emph{$q$-ary code} over  $\mathbb{F}_q$.

The elements of the set $V_n$ are called \emph{words}.
The words that belong to a code are called \emph{codewords}.
We assume that the zero word always belongs to the code unless otherwise specified.
If the codewords of a $q$-ary code form a subspace in the vector space $V_n$,
then the code is called \emph{linear}.

For any pair of words ${\bf x} = (x_1, x_2, \ldots, x_n)$ and ${\bf y} = (y_1, y_2, \ldots, y_n)$ in $V_n$, we define the \emph{Hamming distance} $d({\bf x},{\bf y})$.
Let us set $d({\bf x},{\bf y}) = |\{i \, | \, x_i \neq y_i\}|$.
The \emph{Hamming sphere } of radius $r$ centered at the word ${\bf x}$ is the set $B_r({\bf x}):= \{{\bf y} \in V_n \, | \, d({\bf x},{\bf y}) \leq r \}$.
The \emph{packing radius} $e({\mathcal C})$ of a code ${\mathcal C}$  of length $n$ is the maximum number $e \in \{0, 1, \ldots, n\}$
such  that $B_e({\bf u})\cap B_e({\bf v}) = \varnothing$ for all ${\bf u}, {\bf v} \in {\mathcal C}$, ${\bf u} \neq {\bf v}$.
\emph{The covering radius} $\rho({\mathcal C})$ of a code ${\mathcal C}$ of length $n$ is the minimum number $\rho \in \{0, 1, \ldots, n\}$
such that $\cup_{{\bf c} \in {\mathcal C}} {B_\rho}({\bf c}) = V_n$.

A code ${\mathcal C}$  is called \emph{perfect} if $\rho({\mathcal C})= e({\mathcal C})$
and a code ${\mathcal C}$ is called \emph{quasi-perfect} if $\rho({\mathcal C}) = e({\mathcal C}) + 1$.
If the packing  radius of a perfect  code is equal to $e$,
then the code is called \emph{$e$-perfect}.

The minimum distance between any two different codewords of a code ${\mathcal C}$ is called the \emph{minimum distance} of the code ${\mathcal C}$.
We will use the notation $(n, M, d)_q$ for a $q$-ary code of length $n$, size $M$, and minimum distance $d$.
For a linear $q$-ary code of length $n$, dimension $k$, and  minimum distance $d$, we will use the notation $[n, k, d]_q$.
In the case of binary codes, the index $q$ is omitted.

Perfect  codes attain   the sphere-packing  bound
and are optimal in the sense that they have the largest number of codewords for this length and distance \cite [Ch. 17]{mac}.

All linear  perfect $q$-ary codes over finite fields are known.
These are   Hamming codes that form an infinite family of linear $1$-perfect $q$-ary codes
with parameters $[n, n - m, 3]_q$, where $n = \frac{q^m - 1}{q - 1}$, $m \geq 2$,
and two sporadic Golay codes with parameters $[23, 12, 7]_2$ and $[11, 6, 5]_3$.

In addition to linear perfect $q$-ary codes, there are also non-linear $1$-perfect $q$-ary codes
and trivial perfect codes (a code containing only one codeword, or the whole space, or a binary repetition code of odd length)  \cite [Ch. 6]{mac}.

Non-linear $1$-perfect $q$-ary codes over finite fields form an infinite family of codes with parameters $(n, q^{n - m}, 3)_q$, which are the same as the parameters of Hamming codes.
Binary non-linear $1$-perfect codes were discovered by Vasil'ev \cite{vas}.
A generalization of Vasil'ev's construction to the $q$-ary case was proposed by Lindstr\"{o}m \cite{lin} and Sch\"{o}nheim \cite{sch1}.

Independently, Zinoviev and Leontiev \cite{zin}, as well as Tiet\"{a}v\"{a}inen \cite{tit} proved that over finite fields there are no other perfect $q$-ary codes than those listed above.

Let ${\mathcal C}$  and ${\mathcal D}$  be codes of length $n$ over a field  $\mathbb{F}_q$.
Then ${\mathcal C}$  and ${\mathcal D}$ are \emph{equivalent} if there exist $n$ permutations $\pi_1, \ldots, \pi_n$
of the $q$ elements and a permutation $\sigma$ of the $n$ coordinate positions such that
$(u_1, \ldots, u_n) \in  {\mathcal C}$ if and only if $\sigma(\pi_1(u_1), \ldots, \pi_n(u_n))\in  {\mathcal D}$.

Golay codes are unique up to equivalence \cite {del1}.
Hamming codes, i.e.  linear  $1$-perfect  $q$-ary  codes, are also unique in the sense that any linear
code with parameters $[n, n - m, 3]_q$ is equivalent to the $q$-ary  Hamming code of length $n$.
Denote by $N_q(n)$ the number of nonequivalent $1$-perfect $q$-ary codes of length $n$.
There are well-known lower and upper bounds for  $N_q(n)$:
$$ q^{q^{cn(1 - o(n))}}\leq N_q(n)  \leq  q^{q^{n(1 - o(n))}}, $$
where $c = \frac{1}{q}$.
The lower bound for  $N_q(n)$ was first proved by Vasil'ev \cite{vas} and Lindstr\"{o}m \cite{lin},
and also follows directly from Sch\"{o}nheim's construction \cite{sch1}.
The upper bound for  $N_q(n)$ is trivial.
In \cite{hed3}  it is  proved that  $c  \simeq \frac{2}{q}$ for  $q \geq 4$.
Note the significant gap between the lower  and  upper  bounds.

It is also known that  $N_2(7) = 1$ \cite{zar}, $N_2(15) = 5 983 $  \cite{ost}, $N_3(4) = 1$ \cite{tau}, $N_3(13) \geq 93 241 327$  \cite{shi},
$N_4(5) = 1$ \cite{ald}, $N_5(6) =  N_7(8) = 1$, $N_8(9) = 4$ \cite{kok}.
For $q \geq 9$, $q = p^t$, $t \geq 2$, it is  known that  $N_q(q + 1) \geq 2$ \cite{lin}.

For $q \geq 5$, there exist non-linear $1$-perfect $q$-ary codes that contain the zero word and are equivalent to a linear $1$-perfect $q$-ary code,
i.e. equivalent to the   Hamming code \cite{hed2,rom0}.
The notion of code equivalence, under which the Hamming code equivalence class contains only
the Hamming code and its translates, is considered, for example,  in  \cite{rom1,rom2,rom3}.

The classification  and enumeration  of non-linear $1$-perfect codes over finite fields  is a well-known unsolved problem in coding theory \cite [p. 181]{mac}.

Constructions of $1$-perfect mixed codes were proposed by Sch\"{o}nheim  \cite{sch2},  Herzog and Sch\"{o}nheim  \cite{her2}, and Heden \cite{hed}.
The construction of an infinite family of $2$-perfect mixed codes was proposed by Etzion and Greenberg \cite{etz}.
These are all currently known constructions of perfect mixed codes.

Let $q$ be a power of a prime and $m$ be a positive integer.
In this paper, we propose  a construction of $1$-perfect mixed codes in the Cartesian product $\mathbb{F}_{n} \times  \mathbb{F}_{q}^n$,
where $\mathbb{F}_{n}$ and $\mathbb{F}_{q}$ are finite fields of orders $n = q^m$ and $q$, respectively.
We consider generalized Reed-Muller codes of length $n = q^m$ and   order $(q - 1)m - 2$.
We also consider generalized Reed-Muller codes of length $n = q^m$ and order $(q - 1)m - 1$.
Codes whose parameters are the same as the parameters of generalized Reed-Muller codes are called  \emph{ Reed-Muller-like codes}.
The $q$-ary Reed-Muller-like codes of length $n = q^m$ and order $(q - 1)m - 1$ coincide with $q$-ary distance-2 MDS codes of length $n = q^m$.
The proposed construction is based on partitions
of $q$-ary distance-2 MDS codes of length $n = q^m$ into $q$-ary Reed-Muller-like codes of length $n = q^m$  and order $(q - 1)m - 2$.

In this paper,
we also prove that each  $1$-perfect mixed code in the Cartesian product
$\mathbb{F}_{n} \times  \mathbb{F}_{q}^n$ corresponds
to a certain partition  of a distance-2 MDS code  into Reed-Muller-like codes of order $(q - 1)m - 2$.

Finally, we  prove  that the number of nonequivalent 1-perfect mixed codes
in the product $\mathbb{F}_{n} \times \mathbb{F}_{q}^{n}$ is greater than $q^{q^{cn}}$,
where the constant $c$ satisfies $c < 1$, $n = q^m$ and $m$ is a  sufficiently large positive integer.

The paper is organized as follows.
Section \ref{sec:dif} gives the main definitions and a brief overview of known results.
Section \ref{sec:her} provides an overview of the Herzog and Sch\"{o}nheim  construction of $1$-perfect  mixed codes.
Section \ref{sec:heden}  gives an overview of the Heden  construction of $1$-perfect  mixed codes.
Section \ref{sec:new}  presents a construction of $1$-perfect mixed codes from Reed-Muller-like codes
and shows a connection between  $1$-perfect mixed codes in the product $\mathbb{F}_{n} \times \mathbb{F}_{q}^{n}$
and  partitions of $q$-ary  distance-2 MDS codes of length $n$
into $q$-ary Reed-Muller-like  codes of  length $n$ and order $r = (q - 1)m - 2$.
Section \ref{sec:low}  presents a lower bound for the number of nonequivalent $1$-perfect mixed codes.

\section{Definitions, notation and known results}\label{sec:dif}


In this section, we give some notations, basic definitions, and a brief review of known results.
It is well-known that additive group of $\mathbb{F}_{q}$ is an elementary abelian $p$-group,
where $q = p^t$ with $p$ prime and  $t\geq 1$.
The product $V_n := \mathbb{F}_{1}\times \mathbb{F}_{2} \times \cdots \times \mathbb{F}_{n}$ will be considered as an additive  group   with coordinate-wise addition.
A code is called \emph{additive} if its codewords form an additive subgroup of the group $V_n$.

Let ${\mathcal C}$  and ${\mathcal D}$  be mixed codes of length $n$ in the product $V_n := \mathbb{F}_{1}\times \mathbb{F}_{2} \times \cdots \times \mathbb{F}_{n}$,
where $\mathbb{F}_{i}$ is a finite field of order $q_i$.
Then ${\mathcal C}$  and ${\mathcal D}$ are \emph{equivalent} if  for each  $i = 1, 2,  \ldots,  n$ there exist a permutation $\pi_i$
of $\mathbb{F}_i$ and a permutation $\sigma$ of the $n$ coordinate positions
such that   $(u_1, \ldots, u_n) \in  {\mathcal C}$ if and only if $\sigma(\pi_1(u_1), \ldots, \pi_n(u_n))\in  {\mathcal D}$.

Let ${\bf x } = (x_1, x_2, \ldots, x_n)\in \mathbb{F}_{q}^n$,
$p({\bf x}) := \sum_{i = 1}^n x_i$.
The vector ${\bf x} = (x_1, x_2, \ldots, x_n)\in \mathbb{F}_{q}^n$ is called \emph{even} if $p({\bf x}) = 0$.
A code is called \emph{even} if it has only even codewords.

For a code $\mathcal{C}$ with parameters $(n,M,d)_q$, we define the {\it extended} code $\widehat{\mathcal{C}}$.
Set
$$
\widehat{\mathcal{C}} = \left\{{\bf x} = (x_1, x_2, \ldots,  x_n, x_{n + 1}) \in \mathbb {F}_{q}^{n +1}
\, \, \Big |  \, \,
(x_1, x_2, \ldots,  x_n) \in \mathcal{C}, \, \mbox{and} \,\, p({\bf x}) = 0\right\}.
$$
The extended code $\widehat{\mathcal{C}}$ has parameters $(n + 1,M,\widehat{d}\,)_q$, where $\widehat{d} = d$ or $d + 1$.

Obviously, all $q$-ary  extended $e$-perfect codes are quasi-perfect
and, in particular,  $q$-ary  extended $1$-perfect codes are quasi-perfect  codes.
All extended  binary  $1$-perfect codes are optimal.

The classical Reed-Muller codes are binary codes \cite [Ch. 13]{mac}.
A generalization of Reed-Muller codes to $q$-ary case were proposed by Kasami et al. in \cite{kas}.

Let $\mathbb{F}_{q}[X_1, X_2, \ldots, X_m]$ be the algebra of polynomials in $m$ variables over the field $\mathbb{F}_{q}$.
For a  polynomial $f \in \mathbb{F}_{q}[X_1, X_2, \ldots, X_m]$ we denote by $\text{deg}(f)$   its total degree.
Let $AG(m,q)$  be the $m$-dimensional affine space   over the field $\mathbb{F}_{q}$.
Let $n = q^m$ and the points $P_1, P_2, \ldots, P_n$ of $AG(m,q)$  be arranged in some fixed order.
Let  $r$ be an integer such that $0\leq r \leq (q-1)m$.
Then the {\it generalized Reed-Muller code} \cite{kas} of order $r$  over the field $\mathbb{F}_{q}$ is the following subspace:
\begin{equation*}
{RM}_q(r,m) = \Big\{ \big(f(P_1), f(P_2), \ldots, f(P_n)\big) \,\, \Big| \,\,
f \in \mathbb{F}_{q}[X_1, X_2, \ldots, X_m], \,\, \text{deg}(f) \leq r \Big\}.
\end{equation*}
The code ${RM}_q(r,m)$ has the following parameters \cite[Theorem 5.5]{key},  \cite{kas}:
\begin{enumerate}
\item  the length  is  $q^m$;
\item the dimension is
\begin{equation}\label{eq1}
\sum_{i = 0}^{r} \sum_{k = 0}^{m}{(-1)}^k {m \choose k }{i - kq + m - 1 \choose i - kq };
\end{equation}
\item the minimum distance is
\begin{equation}\label{eq2}
(q - b)q^{m - a - 1},
\end{equation}
where $r = (q - 1)a + b$ and $ 0 \leq b < q - 1$.
\end{enumerate}

For $r \leq (q - 1)m - 1$, the generalized  Reed-Muller code ${RM}_q(r,m)$ is an even code.

The generalized  Reed-Muller code ${RM}_q(r,m)$ of order $r = (q - 1)m - 2$  for $q \geq 3$ have parameters $[n = q^m, n - m - 1, 3]_q$  \cite{rom2}.
In the binary case, the code ${RM}_q(r,m)$ of order $ r = (q - 1)m - 2$ is  the extended Hamming code and have  parameters $[n = 2^m, n - m - 1, 4]$.
Generalized Reed-Muller codes ${RM}_q(r,m)$ of order $r = (q - 1)m - 2$  are  quasi-perfect  \cite{rom2}.
The order of the code dual to the generalized  Reed-Muller code ${RM}_q(r,m)$ of order $r = (q - 1)m - 2$   is 1 \cite {del}.

The generalized  Reed-Muller code ${RM}_q(r,m)$ of order $r = (q - 1)m - 1$  have parameters $[n = q^m, n  - 1, 2]_q$.
Generalized Reed-Muller codes ${RM}_q(r,m)$ of order $r = (q - 1)m - 1$
are the $q$-ary  linear  MDS codes with minimum distance 2
(MDS codes or maximum distance separable codes \cite [Ch. 11]{mac}).
The dual code  of the generalized  Reed-Muller code ${RM}_q(r,m)$ of order $r = (q - 1)m - 1$   is the repetition code.

The  $q$-ary Reed-Muller-like codes of  order $r = (q - 1)m - 1$
are the $q$-ary MDS codes with minimum distance 2.
The $q$-ary MDS codes with minimum distance 2 will be referred to as  $q$-ary distance-2 MDS codes.
There is a one-to-one correspondence between distance-2 MDS codes  of length $n$ over $\mathbb{F}_{q}$
and Latin hypercubes of order $q$ and dimension $n - 1$.
For $q = 2, 3$, Latin hypercubes of order $q$  are  unique \cite[p. 224]{lay}.
Therefore, for $q = 2, 3$, there exists a unique distance-2  MDS code, and this code is linear over $\mathbb{F}_{q}$.
For $q \geq 4$, there are non-linear distance-2  MDS codes  \cite{zin2}.

Two Latin hypercubes of order $q$ and   dimension $n$ are said to be {\it orthogonal} if, when they are superimposed,
fixing any $n - 2$ coordinates gives a Greco-Latin square.
A pair of orthogonal Latin hypercubes form a {\it Graeco-Latin hypercube}.

A non-empty set $Q$ of cardinality $q$ with an $n$-ary operation $f : Q^n \rightarrow Q$ uniquely invertible in each argument
is called an  $n$-ary {\it quasigroup} of order $q$.
The function $f$ will also be called an $n$-ary quasigroup of order $q$.
Latin hypercubes of order $q$ and dimension $n$ are $n$-dimensional multiplication tables of $n$-ary quasigroups of order $q$.

We will denote the concatenation of the words ${\bf u}$ and  ${\bf v}$ by $({\bf u}|{\bf v})$.


\section{Herzog and Sch{\"o}nheim codes }\label{sec:her}


In this section, we will consider  the construction of $1$-perfect mixed codes proposed by Herzog and Sch\"{o}nheim \cite{her}.
\begin{thm} \emph{(Herzog and Sch\"{o}nheim  \cite{her}.)}\label{her1}
Let $G$ be an arbitrary finite abelian group. Let $G_1, G_2, \ldots, G_n$ be a family of subgroups of the group $G$ such that
$G = G_1 \cup G_2 \cup \cdots \cup G_n$ and $G_i \cap G_j = \{0\}$ \, for $i \neq j$. Then the set
$$
{\mathcal C} = \left\{(g_1, g_2, \ldots, g_n)\in G_1  \times G_2  \times \cdots  \times G_n \,  \, \Big| \, \, g_1 + g_2 + \cdots + g_n  = 0 \right\}
$$
is an additive 1-perfect code in the product $G_1 \times G_2 \times \cdots \times G_n$.
\end{thm}

In \cite{her}, Herzog and Sch\"{o}nheim  proved, that if a finite abelian  group $G$ has a partition into subgroups,
then $G$ must be isomorphic to   $G^r(p)$ for some $r$, where  $G(p)$ is the additive group of $\mathbb{F}_{p}$.

We give an example of constructing 1-perfect  codes using  Theorem  \ref{her1}.
Consider binary codes ${\mathcal C}_1 = \{000, 100, 010, 110\}$, ${\mathcal C}_2 = \{000, 101\}$,
${\mathcal C}_3 = \{000, 011\}$, ${\mathcal C}_4 = \{000, 111\}$, ${\mathcal C}_5 = \{000, 001\}$.
These codes form a partition of the space $\mathbb{F}_2^3$.
Therefore, by Theorem \ref{her1}, there exists an additive 1-perfect mixed code ${\mathcal C} \subset \mathbb{F}_4 \times \mathbb{F}_2^4$.
If $\mathbb{F}_4 = \{0, 1, \alpha, \beta\}$ then
${\mathcal C}$ contains the following words:
\begin{center}
$(0, 0, 0, 0, 0), (\alpha, 1, 0, 1, 0), (\beta, 1, 1, 0, 0), (1, 1, 0, 0, 1)$,

$(0, 1, 1, 1, 1), (\alpha, 0, 1, 0, 1), (\beta, 0, 0, 1, 1), (1, 0, 1, 1, 0)$.
\end{center}
It is easy to see that the code ${\mathcal C}$  can be constructed in another way,
namely, by partitioning the set of even words in $\mathbb{F}_2^4$ into extended 1-perfect binary codes of length $4$.

For $m \geq 2$, using partitions of the set of even words of the vector space $\mathbb{F}_2^{2^m}$
into extended 1-perfect binary codes of length ${2^m}$,
we can construct a wide class of 1-perfect mixed codes in $ \mathbb{F}_{2^m} \times \mathbb{F}_{2}^{2^m}$.

Herzog and Sch\"{o}nheim proved the existence of a wide class of additive 1-perfect mixed codes \cite{her2}.
Let $\alpha$ and $m$ be integers such that
$m > \alpha \geq 2$, and let $n = 1 + (q^m - q^\alpha)/(q - 1)$.
Then there is an additive 1-perfect mixed code in
$\mathbb{F}_{q^\alpha} \times \mathbb{F}_{q}^{n - 1}$, see also \cite{wee}.

For $\alpha = m - 1 \geq 2$ the Herzog and Sch\"{o}nheim  code is an additive 1-perfect mixed code in
$\mathbb{F}_{q^{m - 1}} \times \mathbb{F}_{q }^{q^{m - 1}}$.

An additive 1-perfect mixed code can be transformed into a non-additive 1-perfect mixed code
using a bijection $\psi : \mathbb{F}_{q} \rightarrow \mathbb{F}_{q}$,
which is not a group homomorphism \cite{pas}.


\section{Heden codes }\label{sec:heden}


The construction of $1$-perfect codes proposed by Heden in \cite{hed}
allows one to construct new $1$-perfect mixed codes
from already known $1$-perfect mixed codes,
and also allows one to construct $q$-ary $1$-perfect codes from $1 $-perfect mixed codes
and vice versa.

Since the number of words contained in a Hamming sphere in $V_n$ does not depend on the center of the sphere,
we can consider the Hamming sphere without specifying its center.

\begin{thm} \emph{(Heden  \cite{hed}.)}\label{hed1}
Let  $n$ and $m$ be a positive integers, $n > m$.
Let
$$
{\mathcal C}' \subset V_m = \mathbb{F}_{1}\times \mathbb{F}_{2} \times  \cdots \times \mathbb{F}_{m},
$$
$$
{\mathcal C}''\subset V_{n - m } =  \mathbb{F}_{m + 1}\times \mathbb{F}_{ m + 2} \times  \cdots \times \mathbb{F}_{n}.
$$
Let ${\mathcal C}'$, ${\mathcal C}''$ be 1-perfect codes, $\mathbb{F}_{i}$ be a finite field of order $q_i$.
Let ${B''_1}$ be a Hamming sphere of radius $1$ in $V_{n - m }$.
Let $|{B''_1}| = q_{m}$.
Then in the product $$ V_{n - 1} = \mathbb{F}_{1}\times \mathbb{F}_{2} \times \cdots \times \mathbb{F}_{m - 1} \times \mathbb{F}_{m + 1} \times \cdots \times \mathbb{F}_{n}$$
there is a 1-perfect code ${\mathcal C}$ of length $n - 1$.
\end{thm}
\begin{proof}
Consider a partition  ${\mathcal C}_1'', {\mathcal C}_2'', \ldots, {\mathcal C}_{q_m}'' $ of the product $V_{n - m }$ into $1$-perfect codes.
Let $\mathbb{F}_{m} = \{\omega_1, \omega_2, \ldots, \omega_{q_{m}}\}$.
In each codeword from ${\mathcal C}'$, the element  $\omega_i \in \mathbb{F}_{m}$, $i \in \{1, 2, \ldots, q_m\}$,   is replaced by codewords from $ {\mathcal C}''_i$,
i.e. from one codeword of  ${\mathcal C}'$,  containing  $\omega_i$,
we will get   $|{\mathcal C}''_i|$  new words.
The code constructed in this way will be the code ${\mathcal C}$ with packing  $e({\mathcal C}) \geq 1$.
Let us show that the coverage radius $\rho({\mathcal C}) = 1$.
Let ${B'_1}$ and ${B_1}$ be Hamming spheres of radius $1$ in $V_{m}$ and $V_{n - 1}$, respectively.
We have $|{B_1}| = |{B'_1}|$ because $|{B''_1}| = q_{m}$.
Since ${\mathcal C}'$ and ${\mathcal C}''$ are 1-perfect codes,
$$ |{\mathcal C}|\cdot|{B_1}| = |{\mathcal C}'|\cdot|{\mathcal C}''|\cdot|B'_1| = |V_{ m}|\cdot|V_{n - m}|/|{B''_1}| = |V_{n - 1}|. $$
Hence $\rho({\mathcal C}) = e({\mathcal C})$ and ${\mathcal C}$ is a 1-perfect code in $V_{n - 1}$.
\end{proof}

Let us give examples of constructing 1-perfect  codes using  Theorem  \ref{hed1}.
It is known that there are 1-perfect codes in $\mathbb{F}_4^{5}$ and in $\mathbb{F}_2^3$.
A Hamming sphere of radius $1$ in  $\mathbb{F}_2^3$ contains $4$ words.
Successive application of   Theorem \ref{hed1} allows us to construct $1$-perfect mixed codes in
$\mathbb{F}_4^4  \times \mathbb{F}_2^3$,
$\mathbb{F}_4^3  \times \mathbb{F}_2^6$,
$\mathbb{F}_4^2  \times \mathbb{F}_2^{9}$,
$\mathbb{F}_4  \times \mathbb{F}_2^{12}$,
and a 1-perfect binary code in $\mathbb{F}_2^{15}$.
It was  computed in \cite{ost1} that the number of nonequivalent 1-perfect mixed codes in
$\mathbb{F}_4^4  \times \mathbb{F}_2^3$,
$\mathbb{F}_4^3  \times \mathbb{F}_2^6$,
$\mathbb{F}_4^2  \times \mathbb{F}_2^{9}$,
$\mathbb{F}_4  \times \mathbb{F}_2^{12}$
is 1, 4, 39,  and  6483, respectively.

It is known that there are 1-perfect codes in $\mathbb{F}_9^{10}$ and in $\mathbb{F}_3^4$.
A Hamming sphere of radius $1$ in  $\mathbb{F}_3^4$ contains $9$ words.
Successive application of   Theorem \ref{hed1} allows us to construct $1$-perfect mixed codes in
$\mathbb{F}_9^9  \times \mathbb{F}_3^4$,
$\mathbb{F}_9^8  \times \mathbb{F}_3^8$,
$\mathbb{F}_9^7  \times \mathbb{F}_3^{12}$,
$\mathbb{F}_9^6  \times \mathbb{F}_3^{16}$,
$\mathbb{F}_9^5  \times \mathbb{F}_3^{20}$,
$\mathbb{F}_9^4  \times \mathbb{F}_3^{24}$,
$\mathbb{F}_9^3  \times \mathbb{F}_3^{28}$,
$\mathbb{F}_9^2  \times \mathbb{F}_3^{32}$,
$\mathbb{F}_9    \times \mathbb{F}_3^{36}$,
and a 1-perfect ternary code in $\mathbb{F}_3^{40}$.

A Hamming sphere of radius $1$ in a vector space of dimension
$\frac{n - 1}{q - 1}$ over  $\mathbb{F}_q$ contains $n$ words, $n = q^m$.
Applying  Theorem \ref{hed1} to additive 1-perfect mixed Herzog and Sch\"{o}nheim  codes
in $\mathbb{F}_n \times \mathbb{F}_{q}^{n}$
allows us to construct $q$-ary 1-perfect  codes in a vector space of dimension
$n + \frac{n - 1}{q - 1}$.

As noted above, we can construct 1-perfect mixed codes in $\mathbb{F}_{2^m} \times \mathbb{F}_{2}^{2^m}$
using partitions of the set of even words of the vector space $\mathbb{F}_2^{2^m}$
into binary extended 1-perfect  codes of length ${2^m}$,
so in this case  Theorem \ref{hed1}  implies the following theorem.

\begin{thm} \label{hed2}
Let ${\mathcal C}_1', {\mathcal C}_2', \ldots, {\mathcal C}_{2^m}'$
be a partition of the set of even words in the vector space $ \mathbb{F}_2^{ 2^m}$
into binary  extended 1-perfect codes of length ${2^m}$.
Let ${\mathcal C}_1'', {\mathcal C}_2'', \ldots, {\mathcal C}_{2^m}'' $ be a partition of $ \mathbb{F}_2^{2^ m - 1}$
into binary 1-perfect  codes of length ${2^m - 1}$.
Let $\pi$ be a permutation acting on the index set $\{1, 2, \ldots, 2^m\}$.
Then the set
$$
{\mathcal C} = \left\{ \big({\bf u}|{\bf v}\big) \,\, \Big|  \,\, {\bf u} \in {\mathcal C}'_i, \, {\bf v} \in {\mathcal C}''_{\pi(i)}, \,\, i = 1, 2, \ldots, 2^m\right\}
$$
is a  binary 1-perfect code of length $ {2^{m + 1} - 1}$.
\end{thm}
The construction  presented in Theorem \ref{hed2}  was also known earlier, see for example  \cite{liu}.
This construction  is called the
\emph{"concatenation"} or the \emph{"doubling construction"}
and can be viewed as a combinatorial generalization of the well-known
$(\bf u | \bf u + \bf v)$  construction.
The constructions  presented in Theorem \ref{hed1} and in Theorem \ref{hed2}
can be viewed  as examples of the generalized concatenated code constructions  \cite{zin1}.

Heden proved for the first time the existence of non-linear 1-perfect binary codes that are not equivalent to Vasil'ev codes \cite{hed}.
Solov'eva studied non-trivial partitions of the vector space $ \mathbb{F}^n_2$ into Vasil'ev  codes \cite{sol}.
Phelps showed that there are 11 nonequivalent partitions of $ \mathbb{F}_2^{7}$ into Hamming codes of length $7$
and 10 nonequivalent partitions of $\mathbb{F}_2^{8}$ into extended Hamming codes of length $8$ \cite{phe4}.


\section{Construction of 1-perfect mixed codes from  Reed-Muller-like codes }\label{sec:new}


In this section, we propose a construction of $1$-perfect mixed codes in $\mathbb{F}_n \times \mathbb{F}_{q}^{n}$ that is different from the construction of Herzog and Sch\"{o}nheim.

\begin{thm}
\label{constr1}
Let $\mathbb{F}_n = \{\omega_1, \omega_2, \ldots, \omega_n\}$ be a finite field of order $n = q^m$, $m \geq 1$ and
let ${\mathcal C }_1', {\mathcal C}_2', \ldots, {\mathcal C}_{n}' $ be a partition of a $q$-ary  distance-2 MDS code of length $q^m$
into  $q$-ary  Reed-Muller-like codes of length $q^m$ and  order $(q - 1)m - 2$.
Then the set
$$
{\mathcal C} = \left\{ \big({\omega_{i}}|{\bf v}\big) \,\, \Big | \,\, {\omega_{i}}  \in \mathbb{F}_n, \,  {\bf v}  \in  {\mathcal C}'_i, \,\, i = 1, 2, \ldots, n\right\}
$$
for $m \geq 2$ is a 1-perfect mixed code in the product $ \mathbb{F}_{n} \times \mathbb{F}_{q}^{n}$,
and for $m = 1$, $q \geq 3$ is a $q$-ary 1-perfect  code of length $q + 1$.
The converse is also true.
\end{thm}
\begin{proof}
We will assume that $q \geq 3$.
For $q = 2$ Theorem  \ref{constr1} is proved similarly.
Consider the packing radius $e({\mathcal C})$ of the code ${\mathcal C}$.
Any $q$-ary   distance-2 MDS code of length $q^m$ is a $q$-ary Reed-Muller-like code of length $q^m$  and order $ r = (q - 1)m - 1$.
Since the minimum distance of a Reed-Muller-like code  of order $ r = (q - 1)m - 1$ is $2$,
and the minimum distance of a Reed-Muller-like code of order $(q - 1 )m - 2$ is $3$,
we have $e({\mathcal C}) \geq 1$.
Let us show that the covering radius $\rho({\mathcal C}) = 1$.
Let $B_1$ be the Hamming sphere of radius $1$ in the product $ \mathbb{F}_{n} \times \mathbb{F}_{q}^{n}$.
Then
$$
|{RM}_q((q - 1)m - 1,m)|\cdot|B_1| = q^{n - 1}q^{m + 1} = q^{n + m} = |\mathbb{F}_{n} \times  \mathbb{F}_{q}^{n}|.
$$
Hence $\rho({\mathcal C}) = e({\mathcal C})$ and the set ${\mathcal C}$ is a 1-perfect code in the product $ \mathbb{F}_{n} \times \mathbb{F}_{q}^{n}$.

The converse is true because, by Theorem \ref{hed1}, 1-perfect mixed  codes can be embedded in 1-perfect codes,
and in this case Reed-Muller-like codes of order $(q - 1 )m - 2$ attain the Delsarte's linear programming bound \cite{del2}.
For $m = 1$, generalized Reed-Muller codes are extended Reed-Solomon codes   \cite [p. 296]{mac}.
\end{proof}

It is easy to see that
for $n = q^m$,  $m = 1$, $q \geq 3$, there is a one-to-one correspondence between Graeco-Latin hypercubes of order $q$
and dimension $n - 1$ and  partitions of $q$-ary  distance-2 MDS codes of length $n$
into $q$-ary Reed-Muller-like  codes of  length $n$ and order $r = (q - 1)m - 2$.

Since the 1-perfect mixed code ${\mathcal C}$ from   Theorem \ref{constr1} is constructed using partitions,
we can generalize  Theorem  \ref{hed2} to the $q$-ary case.

\begin{thm} \emph{(Romanov \cite{rom1}.)}\label{rom1}
Let $n = q^m$, $m \geq 2$ and ${\mathcal C}_1', {\mathcal C}_2', \ldots, {\mathcal C}_{n}'$
be a partition of a $q$-ary  distance-2 MDS code of length $n = q^m$
into $q$-ary Reed-Muller-like codes of length $n = q^m$ and order $(q - 1)m - 2$.
Let ${\mathcal C}_1'', {\mathcal C}_2'', \ldots, {\mathcal C}_{n}'' $ be a partition of
a vector space of dimension $\frac{n - 1}{q - 1}$ over the field $\mathbb{F}_q$ into $q$-ary 1-perfect  codes  of length $\frac{n - 1}{q - 1}$.
Let $\pi$ be a permutation acting on the index set $\{1, 2, \ldots, n\}$.
Then the set
$$
{\mathcal C} = \left\{ \big({\bf u}|{\bf v}\big) \,\, \Big | \,\, {\bf u}  \in  {\mathcal C}'_i, \,  {\bf v}  \in  {\mathcal C}''_{\pi(i)}, \, \, i = 1, 2, \ldots, n\right\}
$$
is a $q$-ary 1-perfect  code of length $n + \frac{n - 1}{q - 1}$.
\end{thm}

The proof of Theorem  \ref{rom1} proposed in \cite{rom1} is based on a different approach.
Theorem  \ref{rom1} allow us to construct  at least $93 241 327$  nonequivalent $1$-perfect ternary codes of length $13$ \cite{shi}.



\section{Lower bound on the number of nonequivalent 1-perfect mixed codes.}\label{sec:low}


In this section, we construct partitions of $q$-ary distance-2 MDS codes into $q$-ary Reed-Muller-like  codes of   order $(q - 1)m_1 - 2$.
Our construction is based on the method proposed by Phelps in \cite{ph2}.
We also prove a lower bound on the number of nonequivalent   1-perfect mixed codes in the product $ \mathbb{F}_{n} \times \mathbb{F}_{q}^{n}$.

\begin{prop} \label{Pr1}
Let $m_1$ and $m_2$ be  positive integers.
Let ${\mathcal A}^1, {\mathcal A}^2, \ldots, {\mathcal A}^q$ be a partition of $\mathbb{F}_q^{q^{m_1}}$ into $q$-ary   distance-2 MDS codes of length $q^{m_1}$.
Let ${\mathcal B}$  be a $q$-ary distance-2 MDS codes of length $q^{m_2}$.
Let there be a one-to-one correspondence between the superscripts $\{1, 2, \ldots, q\}$ and elements of the field $\mathbb{F}_{q}$.
Then the set
\begin{multline}\label{eq7}
{\mathcal C} =  \Big\{\big({\bf u}_1| {\bf u}_2 | \cdots | {\bf u}_{q^{m_2}} \big)\, \,  \Big|  \, \,
{\bf u}_i \in {\mathcal A}^{v_i}, \,
{\bf v} = (v_1, v_2, \ldots, v_{q^{m_2}}) \in {\mathcal B},
\,\, i = 1, 2, \ldots, {q^{m_2}}\Big\}
\end{multline}
is a $q$-ary   distance-2 MDS code of  length $ n = q^{m_1 + m_2}$.
\end{prop}
\begin{proof}
Let us show that the code ${\mathcal C}$ has the same parameters as the
$q$-ary  distance-2 MDS code of  length $q^{m_1 + m_2}$.

Since the length of the code ${\mathcal B}$ is $q^{m_2}$ and for each
$i = 1, 2,  \ldots, {q^{m_2}}$ the length of the  code  ${\mathcal C}^{v_i}$  is $q^{m_1}$,
it follows from (\ref{eq7}) that the length of the code ${\mathcal C}$ is $q^{m_1 + m_2}$.

Now we will show that the number of codewords in the code ${\mathcal C}$  is $q^{q^{m_1 + \,m_2} - 1}$.

Since $|{\mathcal A}^{v_i}| = q^{q^{m_1}  - 1}$ for all  $i \in \{1, 2, \ldots ,{q^{m_2}}\}$,
it follows from (\ref{eq7})  that for each ${\bf v} \in {\mathcal B}$ one can construct
$$
|{\mathcal A}^{v_i}|^{q^{m_2}} = \big(\,q^{m_1 - 1}\big)^{q^{m_2}} = q^{q^{m_1 + \,m_2} - q^{m_2}}
$$
codewords of the code ${\mathcal C}$.
Since $|{\mathcal B}| = q^{q^{m_2} - 1}$, we have
$$
|{\mathcal C}| = |{\mathcal B}|\,q^{q^{m_1 + \,m_2} - q^{m_2}} =  q^{q^{m_1 + \, m_2}  - 1}.
$$

Next, we will show that the minimum  distance of ${\mathcal C}$ is 2.
Let ${\bf u} = \big({\bf u}_1| {\bf u}_2 | \cdots | {\bf u}_{q^{m_2}} \big)$
and ${\bf u}' = \big({\bf u}'_1| {\bf u}'_2 | \cdots | {\bf u}'_{q^{m_2}} \big)$
be codewords in ${\mathcal C}$.
Then
$$
d({\bf u},{\bf u}') \geq \sum_{i = 1}^{\,{q^{m_2}}} d({\bf u}_i, {\bf u}'_i),
$$
where  words  ${\bf u}_i$ and ${\bf u}'_i$ have length $q^{m_1}$ for  all $i \in \{1, 2, \ldots, q^{m_2}\}$.
The codewords ${\bf u}$, ${\bf u}'$ correspond to the codewords ${\bf v} = (v_1, v_2, \ldots, v_{q^{m_2}})$,
${\bf v}' = (v'_1, v'_2, \ldots, v'_{q^{m_2}})$  of the code ${\mathcal B}$.
If ${\bf u}_i = {\bf u}'_i$ then $v_i = v'_i$. If $v_i \neq v'_i$ then  $d({\bf u}_i, {\bf u}'_i) \geq 1$.
Therefore, if  $d({\bf v}, {\bf v}') \geq 2$, then $d({\bf u}_i,{\bf u}'_i) \geq 1$ for  at least two values of $i$.
Thus, we conclude that if  ${\bf v} \neq {\bf v}'$ then
$$
\sum_{i = 1}^{\,{q^{m_2}}} d({\bf u}_i,{\bf u}'_i) \geq 2.
$$

If  ${\bf v} = {\bf v}'$ then ${\bf u}_i$ and ${\bf u}'_i$
belong to the same element of the partition ${\mathcal A}^1, {\mathcal A}^2, \ldots, {\mathcal A}^q$
for  all $i \in \{1, 2, \ldots, q^{m_2}\}$.
Therefore, if  ${\bf v} = {\bf v}'$ and  ${\bf u}_i \neq {\bf u}'_i$ then $d({\bf u}_i,{\bf u}'_i) \geq 2 $.
\end{proof}

\begin{thm} \label{th6}
Let $m_1$ and $m_2$ be  positive integers.
Let ${\mathcal A}^1, {\mathcal A}^2, \ldots, {\mathcal A}^q$ be a partition of $\mathbb{F}_q^{q^{m_1}}$ into $q$-ary  distance-2 MDS  codes of length $q^{m_1}$.
For each  $k \in \{1, 2, \ldots, q\}$,
let
$$
{\mathcal A}^k_1, {\mathcal A}^k_2, \ldots, {\mathcal A}^k_{q^{m_1}}
$$
be a  partition of ${\mathcal A}^k$ into $q$-ary  Reed-Muller-like  codes of  length $q^{m_1}$ and order $(q - 1)m_1 - 2$.
Let ${\mathcal B}$  be a $q$-ary Reed-Muller-like  code of  length $q^{m_2}$ and order $(q - 1)m_2 - 2$.
For each codeword ${\bf v} \in {\mathcal B}$, we define a $(q^{m_2} - 1)$-ary
quasigroup ${\bf q_v}$ of order $q^{m_1}$.
A quasigroup is defined on the set of indexes $\{1, 2, \ldots, q^{m_1}\}$.
Let there be a one-to-one correspondence between the superscripts $\{1, 2, \ldots, q\}$ and elements of the field $\mathbb{F}_{q}$.
Then the set
\begin{multline}\label{eq8}
{\mathcal C} = \Bigg\{\big({\bf u}_1| {\bf u}_2 | \cdots | {\bf u}_{q^{m_2}} \big)\, \,  \Big|  \, \,
{\bf u}_i \in {\mathcal A}_{j_i}^{v_i}, \\
{\bf v} = (v_1, v_2, \ldots, v_{q^{m_2}}) \in {\mathcal B},
\,\, j_1 =  {\bf q_v}(j_2, j_3, \ldots, j_{q^{m_2}}), \\
j_i \in \{ 1, 2, \ldots,   {q^{m_1}} \},
\,\, i = 1, 2, \ldots, {q^{m_2}} \Bigg\}
\end{multline}
is a $q$-ary  Reed-Muller-like  code of  length $ q^{m_1 + m_2}$ and order $(q - 1)(m_1 + m_2) - 2$.

Assume that the code ${\mathcal B}$  belongs to a partition of a $q$-ary  distance-2 MDS code of length $q^{m_2}$
into $q$-ary Reed-Muller-like  codes of  length $q^{m_2}$ and order $(q - 1)m_2 - 2$.
Then the code ${\mathcal C}$ also belongs to a partition of some  $q$-ary  distance-2 MDS code of length $q^{m_1 + m_2}$
into $q$-ary Reed-Muller-like  codes of  length $q^{m_1 + m_2}$ and order $(q - 1)(m_1 + m_2) - 2$.
\end{thm}
\begin{proof}
First we show that the code ${\mathcal C}$  has the same parameters
as the generalized  Reed-Muller codes of  length $q^{m_1 + m_2}$ and order $(q - 1)(m_1 + m_2) - 2$.

Since the length of the code ${\mathcal B}$ is $q^{m_2}$ and for each
$j_i = 1, 2, \ldots,   {q^{m_1}}$, $ i = 1, 2,  \ldots, {q^{m_2}}$ the length of the  code  ${\mathcal A}_{j_i}^{v_i}$  is $q^{m_1}$,
it follows from (\ref{eq8}) that the length of the code ${\mathcal C}$ is $q^{m_1 + m_2}$.

Now we will show that the number of codewords in the code ${\mathcal C}$  is $q^{q^{m_1 + \,m_2}- (m_1 + m_2) - 1}$.
For each  $k \in \{1, 2, \ldots, q\}$,  the code ${\mathcal A}^k$
is  $q$-ary   distance-2 MDS code of length $q^{m_1}$  contains $q^{q^{m_1} - 1}$ codewords.
Each ${\mathcal A}^k$ is divided into $q^{m_1}$ subcodes ${\mathcal A}^k_1, {\mathcal A}^k_2, \ldots, {\mathcal A}^k_{q^{m_1}}$.
Consequently,
$$
|{\mathcal A}_{j_i}^{v_i}|\,q^{m_1} = q^{q^{m_1} - 1},
$$
which implies that $|{\mathcal A}_{j_i}^{v_i}| = q^{q^{m_1} - m_1 - 1}$ for $j_i \in \{1, 2, \ldots, {q^{m_1}}\}$ and $i = 1, 2, \ldots ,{q^{m_2}}$.
It follows from (\ref{eq8})  that for each ${\bf v} \in {\mathcal B}$ one can construct
$$
|{\mathcal A}_{j_1}^{v_1}|\cdot|{\mathcal A}_{j_i}^{v_i}|^{q^{m_2} - 1}\big(\,q^{m_1}\big)^{q^{m_2} - 1} =
|{\mathcal A}_{j_1}^{v_1}|\cdot\Big(|{\mathcal A}_{j_i}^{v_i}|\,q^{m_1}\Big)^{q^{m_2} - 1} =
q^{q^{m_1 + \,m_2} - q^{m_2} - m_1 }
$$
codewords of the code ${\mathcal C}$.
Since $|{\mathcal B}| = q^{q^{m_2} - m_2 - 1}$, we have
$$
|{\mathcal C}| = |{\mathcal B}|\,q^{q^{m_1 + \,m_2} - q^{m_2} - m_1 } =  q^{q^{m_1 + \, m_2} - (m_1 + \,m_2) - 1}.
$$

Next, we will show that the minimum  distance of ${\mathcal C}$ is 4 for $q = 2$ and 3 for $q \geq 3$.
Let ${\bf u} = \big({\bf u}_1| {\bf u}_2 | \cdots | {\bf u}_{q^{m_2}} \big)$
and ${\bf u}' = \big({\bf u}'_1| {\bf u}'_2 | \cdots | {\bf u}'_{q^{m_2}} \big)$
be codewords in ${\mathcal C}$.
Then
$$
d({\bf u},{\bf u}') \geq \sum_{i = 1}^{\,{q^{m_2}}} d({\bf u}_i, {\bf u}'_i),
$$
where  words  ${\bf u}_i$ and ${\bf u}'_i$ have length $q^{m_1}$ for  all $i \in \{1, 2, \ldots, q^{m_2}\}$.
The codewords ${\bf u}$, ${\bf u}'$ correspond to the codewords ${\bf v} = (v_1, v_2, \ldots, v_{q^{m_2}})$,
${\bf v}' = (v'_1, v'_2, \ldots, v'_{q^{m_2}})$  of the code ${\mathcal B}$.
If ${\bf u}_i = {\bf u}'_i$ then $v_i = v'_i$. If $v_i \neq v'_i$ then  $d({\bf u}_i, {\bf u}'_i) \geq 1$.
Therefore, if  $q \geq 3$ and $d({\bf v}, {\bf v}') \geq 3$, then $d({\bf u}_i,{\bf u}'_i) \geq 1$ for  at least three values of $i$.
If  $q = 2$ and $d({\bf v}, {\bf v}') \geq 4$, then $d({\bf u}_i, {\bf u}'_i) \geq 1$ for  at least four values of $i$.
Thus, we conclude that if  ${\bf v} \neq {\bf v}'$ then
$$
\sum_{i = 1}^{\,{q^{m_2}}} d({\bf u}_i,{\bf u}'_i) \geq 4  \, \, \,  \mbox{for} \, \, \,  q = 2,
$$
$$
\sum_{i = 1}^{\,{q^{m_2}}} d({\bf u}_i,{\bf u}'_i) \geq 3  \, \, \, \mbox{for} \, \, \,  q \geq 3.
$$

If  ${\bf v} = {\bf v}'$ then ${\bf u}_i$ and ${\bf u}'_i$
belong to the same element of the partition ${\mathcal A}^1, {\mathcal A}^2, \ldots, {\mathcal A}^q$
for  all $i \in \{1, 2, \ldots, q^{m_2}\}$.
Therefore, if  ${\bf v} = {\bf v}'$ and  ${\bf u}_i \neq {\bf u}'_i$ then $d({\bf u}_i,{\bf u}'_i) \geq 2 $.
Assume that  ${\bf u}_i \in {\mathcal A}_{j_i}^{v_i}$ and ${\bf u}'_i \in {\mathcal A}_{j_x}^{v_i}$ where $i = 1, 2, \ldots, q^{m_2}$.
Then  the equality $d({\bf u}_i, {\bf u}'_i) = 0$ means that $j_i = j_x$.
Since  ${\bf j} = (j_1, j_2, \ldots , j_{q^{m_2}})$ and ${\bf x} = (x_1, x_2, \ldots , x_{q^{m_2}})$ are defined by a  quasigroup,
we conclude that if ${\bf j} \neq {\bf x}$ then $d({\bf u}_i,{\bf u}'_i) \geq 2 $ for at least two values of $i$.
Hence, $d({\bf u},{\bf u}') \geq 4$ except for the case where ${\bf j} = {\bf x}$.
However, in this case,  ${\bf u}_i \neq {\bf u}'_i$ means
that ${\bf u}_i$ and ${\bf u}'_i$ belong to the same element of the partition of ${\mathcal A}^k$ into   Reed-Muller-like  codes.

Let us now prove the second part of Theorem \ref{th6}.
Assume that the code ${\mathcal B}$  belongs to a partition of a $q$-ary  distance-2 MDS code of length $q^{m_2}$
into $q$-ary Reed-Muller-like  codes of  length $q^{m_2}$ and order $(q - 1)m_2 - 2$.

For $q \geq 3$, the generalized  Reed-Muller code ${RM}_q(r,m)$ of order $ r = (q - 1)m - 2$
and length $n = q^m$  has parameters $[n = q^m, n - m - 1, 3]_q$  \cite{rom2}.
In the binary case, the code ${RM}_q(r,m)$ of order $r = (q - 1)m - 2$  has  parameters $[n = 2^m, n - m - 1, 4]$.
The $q$-ary   distance-2 MDS code of length $n = q^m$ has parameters $(n = q^m, q^{n  - 1}, 2)_q$.
Therefore, a partition of a $q$-ary  distance-2 MDS code of length $q^{m_1}$
into $q$-ary Reed-Muller-like  codes of  length $q^{m_1}$ and order $(q - 1)m_2 - 2$
contains $q^{m_1}$ elements.
Consequently, each element of the partition  of a $q$-ary  distance-2 MDS code of length $q^{m_2}$
into $q$-ary Reed-Muller-like  codes of  length $q^{m_2}$ and order $(q - 1)m_2 - 2$
allows us to construct $q^{m_1}$ pairwise disjoint codes with
parameters of a $q$-ary Reed-Muller-like  codes of  length $q^{m_1 + m_2}$ and order $(q - 1)(m_1 + m_2) - 2$.
Any partition of a $q$-ary  distance-2 MDS code of length $q^{m_2}$
into $q$-ary Reed-Muller-like  codes of  length $q^{m_2}$ and order $(q - 1)m_2 - 2$ contains $q^{m_2}$ elements.
Therefore, we can construct $q^{m_1 + m_2}$ pairwise
disjoint $q$-ary Reed-Muller-like  codes of  length $q^{m_1 + m_2}$ and order $(q - 1)(m_1 + m_2) - 2$.
By Proposition \ref{Pr1}, these codes form a partition of some $q$-ary  distance-2 MDS code of length $q^{m_1 + m_2}$.
\end{proof}

\begin{thm} \label{th7}
The number of different partitions of some $q$-ary  distance-2 MDS code of length $q^{m_1 + m_2}$
into $q$-ary Reed-Muller-like  codes of  length $q^{m_1 + m_2}$ and order $(q - 1)(m_1 + m_2) - 2$ is greater than
\begin{equation}\label{eq313}
\big(\, q^{m_1}\big)^{\Big(\big(q^{m_1}\big)^2 - \big(q^{m_1}\big)\Big)\big(q^{m_2} - 2\big)q^{{q^{m_2}} - {m_2} - 1}},
\end{equation}
where $q^{m_1 + m_2}$ is a  sufficiently large number.
\end{thm}
\begin{proof}
Let $Q\big(q^{m_2} - 1, q^{m_1}\big)$ denote the set of all $(q^{m_2} - 1)$-ary quasigroups  of order $q^{m_1}$.
It was shown in \cite{ph2} that for sufficiently large $m_1$ the following inequality holds:
\begin{equation}\label{eq314}
\big|Q\big(q^{m_2} - 1, q^{m_1}\big)\big|  \geq
\big(\, q^{m_1}\big)^{\Big(\big(q^{m_1}\big)^2 - \big(q^{m_1}\big)\Big)\big(q^{m_2} - 2\big)}.
\end{equation}
Next, consider the construction (\ref{eq8}) of $q$-ary Reed-Muller-like  codes of  length $q^{m_1 + m_2}$ and order $(q - 1)(m_1 + m_2) - 2$.
Since for each codeword ${\bf v}\in {\mathcal B}$  we can choose any quasigroup
$$
{\bf q_v} \in Q\big(q^{m_2} - 1, q^{m_1}\big).
$$
Then starting from one fixed partition ${\mathcal A}^k_1, {\mathcal A}^k_2, \ldots, {\mathcal A}^k_{q^{m_1}}$ of ${\mathcal A}^k$, we can construct more than
\begin{equation}\label{eq315}
\big|Q\big(q^{m_2} - 1, q^{m_1}\big)\big|^{|\mathcal B|}
\end{equation}
different $q$-ary Reed-Muller-like  codes of  length $q^{m_1 + m_2}$ and order $(q - 1)(m_1 + m_2) - 2$.
We have that
$$
|\mathcal B| = q^{{q^{m_2}} - {m_2} - 1}.
$$
Therefore
\begin{equation}\label{eq316}
\big|Q\big(q^{m_2} - 1, q^{m_1}\big)\big|^{|\mathcal B|}  \geq
\big(\, q^{m_1}\big)^{\Big(\big(q^{m_1}\big)^2 - \big(q^{m_1}\big)\Big)\big(q^{m_2} - 2\big)q^{{q^{m_2}} - {m_2} - 1}}.
\end{equation}

It was shown in \cite{ph2} that inequality (\ref{eq314}) is also valid for quasigroups $\bf q$ such that
$$
{\bf q}(1, 1, . . . , 1) = 1.
$$
Assume that the zero word always belongs to  ${\mathcal A}^1$ in the  partition ${\mathcal A}^1, {\mathcal A}^2, \ldots, {\mathcal A}^q$.
Since we assume that the zero word always belongs to the code,
it follows from (\ref{eq8}) that the zero word  belongs to the $\mathcal C$.
Since each code ${\mathcal C}$ constructed using Theorem 5 contains the zero word
and can belong to a partition of some  $q$-ary   distance-2 MDS code of length $q^{m_1 + m_2}$
into $q$-ary Reed-Muller-like  codes of  length $q^{m_1 + m_2}$ and order $(q - 1)(m_1 + m_2) - 2$,
we conclude that  Theorem \ref{th7} holds.
\end{proof}

\begin{cor}
\label{Cl:2}
The number of nonequivalent 1-perfect mixed codes in the product $\mathbb{F}_{n} \times \mathbb{F}_{q}^{n}$
is greater than $q^{q^{cn}}$,
where the constant $c$ satisfies $c < 1$, $n = q^m$ and $m$ is a  sufficiently large positive integer.
\end{cor}
\begin{proof}
The set of different $1$-perfect mixed codes in the Cartesian product
$\mathbb{F}_{n} \times  \mathbb{F}_{q}^n$, is partitioned into equivalence classes.
An arbitrary equivalence class in this partition contains at most $n!n!(q!)^n$
different $1$-perfect mixed codes.
Since $n = q^m$, we have
\begin{equation}\label{eq317}
n!n!(q!)^n \leq n^{n + 1}n^{n + 1}q^{(q + 1)n} = n^{2(n + 1)}q^{(q + 1)n} = q^{2m(q^m + 1)}q^{(q + 1)q^m}.
\end{equation}
Validity of Corollary \ref{Cl:2} follows from  Theorems\, \ref{constr1}, \ref{th7} and from equation  (\ref{eq317}).
\end{proof}





\end{document}